\documentclass[12pt]{amsart}

\usepackage{upgreek, amssymb, amsmath}
\usepackage[margin=3.5 cm]{geometry}
\usepackage{graphicx}
\usepackage{todonotes}
\usepackage{color}
\usepackage{subfigure}
\newtheorem{theorem}{Theorem}[section]

\newtheorem{prop}{Proposition}[section]

\newtheorem{lemma}[theorem]{Lemma}

\theoremstyle{definition}

\theoremstyle{remark}
\newtheorem{remark}[theorem]{Remark}

\numberwithin{equation}{section}

\DeclareMathOperator{\im}{Im}
\newcommand{\ud}{\mathrm{d}}

\begin{document}

\title[]{Dominated splittings and the spectrum of quasi-periodic Jacobi operators}

\author{C. A. Marx}
\address{Department of Mathematics, Oberlin College}
\email{cmarx@oberlin.edu}





\begin{abstract}
We prove that the resolvent set of any, possibly singular, quasi periodic Jacobi operator is characterized as the set of all energies whose associated Jacobi cocycles induce a dominated splitting. This extends a well-known result by R. A. Johnson for Schr\"odinger operators. 
\end{abstract}

\maketitle

\section{Introduction}

The purpose of this article is to give a dynamical characterization of the spectrum of quasi-periodic Jacobi operators (QPJ). To this end, let $\alpha \in \mathbb{R}^d$ be fixed with components linearly independent over $\mathbb{Q}$ and let $c,v: \mathbb{T}^d \to \mathbb{C}$ be continuous satisfying $v(\mathbb{T}^d) \subseteq \mathbb{R}$. A {\em{quasi-periodic Jacobi operator}} is a family of bounded self-adjoint operators on $\mathit{l}^2(\mathbb{Z})$ of the form,
\begin{equation} \label{eq_jacobiop}
[H_{x} \psi]_n = \overline{c(\mathrm{T}^{n-1}x)} \psi_{n-1} + c(\mathrm{T}^n x) \psi_{n+1} + v(\mathrm{T}^n x) \psi_n ~\mbox{, for $x \in \mathbb{T}^d$,}
\end{equation}
generated upon evaluation of $c,v$ along the trajectories of $\mathrm{T}: \mathbb{T}^d \to \mathbb{T}^d$, $x \mapsto x + \alpha$. Let $\mu$ denote the Haar probability measure on $\mathbb{T}^d$. We are assuming $\log \vert c \vert \in L^1(\mathbb{T}^d, \ud \mu)$, as is commonly done. To simplify notation, we will also set $X:=\mathbb{T}^d$.

Motivated by the now famous almost Mathieu operator where for $d=1$, $v(x) = 2\lambda \cos(2 \pi x)$, $\lambda \in \mathbb{R}$, and $c \equiv 1$, most of the literature on QPJ has so far focussed on the special case $c \equiv 1$, commonly known as quasi-periodic {\em{Schr\"odinger operators}} (QPS). In recent years, a dynamical systems approach to the spectral theory of the latter proved to be extremely fruitful. In particular, this approach allowed for a more global picture of the spectral properties of QPS \cite{Avila_globalthy_published}. 

The relevant dynamical system for QPS are Schr\"odinger cocycles, quasi-periodic $SL_2(\mathbb{R})$-cocycles whose iteration generates solutions to the finite difference equation, $H_x \psi = E \psi$. More generally, a (continuous) {\em{$M_2(\mathbb{C})$-cocycle}} is a dynamical system on $X \times \mathbb{C}^2$ induced by $\mathrm{T}$ and a matrix-valued function $D \in \mathcal{C}(X,M_2(\mathbb{C}))$, defined by $(x,v) \mapsto (\mathrm{T}x, D(x) v)$. We will denote this cocycle map by the pair $(\mathrm{T},D)$.

One fundamental ingredient of above-mentioned dynamical approach to the spectral theory of QPS is a theorem due to R. A. Johnson \cite{Johnson_1986} which characterizes the spectrum in terms of Schr\"odinger cocycles. More specifically, it is shown in \cite{Johnson_1986} that the resolvent set of a QPS is determined by uniformly hyperbolic dynamics of the Schr\"odinger cocycles. We mention that Johnson's theorem has recently been generalized to discrete long-range Schr\"odinger operators \cite{HaroPuig_2013}. The goal of this paper is to find an appropriate extension of Johnson's theorem to QPJ. 

The main problem in extending Johnson's theorem to the more general Jacobi setting is that, whereas Schr\"odinger cocycles are $SL(2,\mathbb{R})$, the relevant cocycles for QPJ are in general not even invertible. In this context, we call a cocycle $(\mathrm{T},D)$ {\em{singular}} if $\det D(x_0) = 0$ for some $x_0 \in X$. For QPJ, singular cocycles automatically arise once the sampling function $c$ has zeros, in which case (\ref{eq_jacobiop}) is called a {\em{singular QPJ}}. 

We mention that QPJ originated in solid states physics, where both $c,v$ are trigonometric polynomials, hence the possibility of $c$ having zeros cannot be excluded in general. For instance, one prominent example of a QPJ relevant in physics is given by extended Harper's model, where for $d=1$ the sampling functions are given by $v(x) = 2 \cos(2 \pi x)$ and $c(x) = \lambda_1 \mathrm{e}^{-2 \pi i (x+ \alpha/2)}  + \lambda_2 + \lambda_3 \mathrm{e}^{2 \pi i (x + \alpha/2)}$ with $\lambda_j \in \mathbb{R}$, $1 \leq j \leq 3$. Proposed by D. J. Thouless in context with the integer quantum Hall effect \cite{Thouless_1983}, extended Harper's model constitutes a singular QPJ for a large set of coupling parameters $(\lambda_1, \lambda_2, \lambda_3)$. Even though several recent works on the spectral theory of QPJ have started to account for singularity (many of which motivated by extended Harper's model) \cite{JitomirskayaMarx_2012, KTao_2011, KTao_2012, AvilaJitomirskayaSadel_2013, BinderVoda_2013, BinderVoda_2014, JitomirskayaMarx_preprint_2014}, an extension of Johnson's theorem to (possibly singular) QPJ has so far been missing. 

The presence of singular cocycles obviously requires a dynamical framework different from uniform hyperbolicity. Recent work on the continuity and positivity of the Lyapunov exponent for Jacobi operators \cite{JitomirskayaMarx_2012, AvilaJitomirskayaSadel_2013, DuarteKlein_2013_1, DuarteKlein_2013_2} indicated the notion ``dominated splitting'' as an appropriate analogue of uniform hyperbolicity, suitable when passing from the Schr\"odinger to the general Jacobi setting. 

A cocycle $(\mathrm{T},D)$ is said to induce a {\em{dominated splitting}} (write $(\mathrm{T},D) \in \mathcal{DS}$) if there exists $N \in \mathbb{N}$ and a continuous, non-trivial splitting of $\mathbb{C}^2 = S_x^{(1)} \oplus S_x^{(2)}$ satisfying $D_N(x) S_x^{(j)} \subseteq S_{\mathrm{T}^N x}^{(j)}$, $1 \leq j \leq 2$, which exhibits uniform domination in the sense,
\begin{equation} \label{eq_defDS}
\dfrac{\Vert D_N(x) v_1 \Vert}{\Vert v_1 \Vert} > \dfrac{\Vert D_N(x) v_2 \Vert}{\Vert v_2 \Vert} ~\mbox{, all $x \in X$,}
\end{equation}
for all $v_j \in S_x^{(j)} \setminus\{0\}$. Here, for $n \in \mathbb{N}$, $D_n(x):= D(\mathrm{T}^{n-1} x) \dots D(x)$ denotes the $n$-th iterate of $(\mathrm{T},D)$ on the fibers, where $D_0(x):=\mathrm{I}$. $\mathcal{DS}$ is a very ``robust property,'' e.g. it is well known \cite{Ruelle_1979} that cocycles inducing a dominated splitting are open in $\mathcal{C}(X,M_2(\mathbb{C}))$, accompanied by continuity, even real analyticity of the (top) {\em{Lyapunov exponent}} (LE) ,
\begin{equation}
L(\mathrm{T},D):= \lim_{n \to +\infty} \frac{1}{n} \int \log \Vert D_n(x) \Vert \ud \mu(x) ~\mbox{.}
\end{equation}

$\mathcal{DS}$ specializes to uniform hyperbolicity ($\mathcal{UH}$) if the matrix cocycles are unimodular. More generally, for any {\em{non-}}singular cocycle, $(\mathrm{T},D) \in \mathcal{DS}$ if and only if $(\mathrm{T},\frac{D}{\sqrt{\mathrm{det} D}}) \in \mathcal{UH}$. The advantage of the notion $\mathcal{DS}$ is however that it makes sense for {\em{both}} singular and non-singular cocycles.

Another feature not present for QPS, is that the relevant cocycles for QPJ ({\em{Jacobi cocycles}}) are not unique. For instance, one possible choice for Jacobi cocycles is given by
\begin{equation} \label{eq_defA}
A^E(x) = \begin{pmatrix} E-v(x) & -\overline{c(\mathrm{T}^{-1} x )} \\
                                                                      c(x)     &  0       \end{pmatrix} ~\mbox{,}
\end{equation}
where $E \in \mathbb{C}$ is the spectral parameter. There are however alternative choices, which, depending on the problem in mind, may be more advantageous. We emphasize that all these choices share that they are singular precisely if $c$ has zeros. We will comment more on the flexibility in the choice of Jacobi cocycles in Sec. \ref{sec_jacobicocycles}. Our results account for this flexibility, and give a dynamical characterization of the spectrum of QPJ applicable for the different choices (cf Theorem \ref{thm_main_ammended}).

To formulate our main result, we recall that for a QPJ, minimality of $\mathrm{T}$ implies that the spectrum of the operators $H_x$ is constant in $x \in X$; we denote this set by $\Sigma$, a compact subset of $\mathbb{R}$.
\begin{theorem} \label{thm_main}
Let $\alpha \in \mathbb{R}^d$ be fixed with components linearly independent over $\mathbb{Q}$. Consider the (possibly singular) quasi-periodic Jacobi operator with frequency $\alpha$, defined in (\ref{eq_jacobiop}), i.e. $\mathrm{T}$ is the translation by $\alpha$ on $\mathbb{T}^d$, and $c,v \in \mathcal{C}(\mathbb{T}^d,\mathbb{C})$ with $v$ real valued and $c \in L^1(\mathbb{T}^d, \ud \mu)$. Then, for the Jacobi cocycle defined in (\ref{eq_defA}), one has
\begin{eqnarray}
\Sigma = \{ E : (\mathrm{T}, A^E) \not \in \mathcal{DS} \} ~\mbox{.} \label{eq_thm_1}
\end{eqnarray} 
\end{theorem}
\begin{remark} \label{rem_thm_main}
\begin{itemize}
\item[(i)] An analogous statement holds for the Jacobi cocycles alternative to (\ref{eq_defA}) which will be described in Sec. \ref{sec_jacobicocycles}, cf. Theorem \ref{thm_main_ammended}. 
\item[(ii)] Our proof of Theorem \ref{thm_main} does in fact not depend on the specifics of the background dynamics induced by translations on $\mathbb{T}^d$. Even though our main motivation for this work are QPJ, the argument we present applies to any uniquely ergodic, invertible map $\mathrm{T}$ on a compact Hausdorff space $X$, where the $\mathrm{T}$-invariant probability measure $\mu$ is topological, i.e. positive on open sets, and the associated Jacobi operators satisfy that $x \mapsto H_x$ is continuous in the operator norm. Here, the latter condition, is needed for Lemma \ref{lem_contisections}. Note that for a compact space X, unique ergodicity with respect to a topological (invariant) measure is equivalent to unique ergodicity and minimality, the latter of which is also known as {\em{strict ergodicity}}.  In particular, Theorem \ref{thm_main} holds for all {\em{almost periodic Jacobi operators}}, i.e. operators of the form (\ref{eq_jacobiop}) where $X$ is a compact topological group, $\mu$ is its Haar probability measure, and $\mathrm{T}$ is translation by a fixed element in $X$ whose orbit is dense.
\item[(iii)] We mention that Johnson's original result extends to Schr\"odinger operators with background dynamics given by a minimal transformation $\mathrm{T}$ on a compact space $X$. Our proof of Theorem \ref{thm_main} requires unique ergodicity of $\mathrm{T}$ (due to Proposition \ref{prop_unifcontractcond}), we however believe that the result should also hold true for the case of merely minimal $\mathrm{T}$.
\end{itemize}
\end{remark}

\vspace{-0.1 cm}

From a dynamical point of view the most interesting aspect of Theorem \ref{thm_main} is that it implies domination outside the spectrum. In particular, complexifying the energy generates $\mathcal{DS}$ with all the ``nice'' properties such dynamics entails. This leads to a dynamical point of view of {\em{Kotani theory}}, which expressed from the point of view of Theorem \ref{thm_main}, studies the limiting properties of the invariant sections as $\im E \to 0+$. We mention that such a dynamical reformulation of Kotani theory has played an important role in a dynamical description of the absolutely continuous spectrum of QPS \cite{AvilaKrikorian_2006, AvilaFayadKrikorian_2011, Avila_prep_ARC_1}, see also \cite{AvilaKrikorian_2013} for an even more general perspective (``monotonic cocycles'').

We structure the paper as follows. Sec. \ref{sec_jacobicocycles} briefly recalls the relation of the Jacobi cocycles to the solutions of the finite difference equations, $H_x \psi = E \psi$. We will also describe alternative choices for Jacobi cocycles appearing in the literature, for which Theorem \ref{thm_main} holds in an analogous form (see Theorem \ref{thm_main_ammended}). 

As mentioned earlier, the main point of this article is to account for {\em{singular}} Jacobi operators. For non-singular operators, Theorem \ref{thm_main} could be obtained by simple adaptions of the proof for Schr\"odinger operators, which we present in Sec. \ref{sec_nonsingularop}.

Sec. \ref{sec_domination} forms the main part of the article, and is devoted to proving $\mathcal{DS}$ outside the spectrum. The latter is done by verifying a {\em{cone condition.}} We outline the strategy in Sec. \ref{subsec_domsplit}, the proof is carried out in Sec. \ref{subsec_conecond}. One noteworthy aspect of our proof is that it explicitly reveals the spectral theoretic meaning of the dynamical quantities involved. The invariant sections of the splitting are shown to be given in terms of the Weyl m-functions ($m_\pm$), with $m_-$ giving rise to the dominating section, cf. (\ref{eq_defsminus}). The key estimate, which verifies the cone condition, is obtained in Proposition \ref{prop_unifcontractcond}. It shows that the derivative of iterates of the Jacobi cocycle along the dominating section decays exponentially in the number of iterates, with a decay rate given by the Lyapunov exponent of the QPJ. Moreover, the angle between the invariant sections of the splitting is shown to be determined by the inverse of the Green's function of the QPJ, cf. (\ref{eq_greensfunctransv_1})-(\ref{eq_greensfunctransv}).

We conclude with Sec. \ref{sec_finishingup} where we show that $\mathcal{DS}$ cannot occur on the spectrum.

\vspace{0.1 cm}
{\bf{Acknowledgment:}} I wish to thank Barry Simon for pointing me to the alternative transfer matrices for Jacobi operators presented in Sec. \ref{sec_jacobicocycles}; this lead to the amended version of Theorem \ref{thm_main}, given in Theorem \ref{thm_main_ammended}.

\vspace{0.1 cm}
{\bf{Concluding remarks:}} {\bf{This is an updated version of the original article which was published in Nonlinearity in 2014 \cite{CMarx_publishedpaper}.}} In August of 2020, it was brought to our attention that Remark 1.2 (ii) in the published version, which concerned the extension of the results of this article from quasi-periodic to the more general situation of almost periodic operators, contained an incorrect statement. Upon correcting this mistake, we realized that the main result of this paper, Theorem \ref{thm_main}, which is formulated and proven for quasi-periodic Jacobi operators, in fact extends to all ergodic Jacobi operators where the base dynamics is induced by a strictly ergodic, invertible map $\mathrm{T}$ on a compact Hausdorff space $X$ for which the associated Jacobi operators satisfy that $x \mapsto H_x$ is continuous in the operator norm. In particular, this allowed to drop the hypothesis of connectedness of the space $X$ as stated in Remark 1.2 (ii) of the original paper. The latter improvement implies that the main results of this article apply to {\em{all}} almost periodic Jacobi operators, irrespective of the base dynamics. 

We would like to thank Jake Fillman for alerting us to the incorrect statement in Remark 1.2 (ii) of the original article, which, ultimately led to the improved result in this updated version. Compared to the original article in \cite{CMarx_publishedpaper}, the only essential changes appear in the proof of Proposition \ref{claim_johnson} of Sec. \ref{sec_finishingup}, where we realized that in equation (5.5) of the published and the present version of the article, a mere upper bound is sufficient for the proving the claim of the proposition. We also updated the list of references.

\section{Jacobi cocycles} \label{sec_jacobicocycles}

Fixing the spectral parameter $E \in \mathbb{C}$, solving $H_x \psi = E \psi$ over $\mathbb{C}^\mathbb{Z}$ can be formulated as iteration of the {\em{measurable}} cocycle\footnote{One can weaken the definition of a (continuous) cocycle, requiring the matrix valued function $D:X \to M_2(\mathbb{C})$ to only be measurable with $\log_+ \Vert D(.) \Vert \in L^1(X, \ud \mu)$, in which case $(\mathrm{T},D)$ is called a {\em{measurable cocycle}}.}  $(\mathrm{T}, B^E)$ applied to a given initial condition $(\begin{smallmatrix} \psi_0 \\ \psi_{-1}\end{smallmatrix})$ for $\psi$,
\begin{equation} \label{eq_B}
B_n^E(x) \begin{pmatrix} \psi_{0} \\ \psi_{-1}  \end{pmatrix} = \begin{pmatrix} \psi_n \\ \psi_{n-1}  \end{pmatrix} ~\mbox{,}
\end{equation}
where $B^E(x) := \frac{1}{c(x)} A^E(x)$ and $A^E$ is given in (\ref{eq_defA}). This iterative procedure is a consequence of the second order difference nature of Jacobi operators. In spectral theory, it is better known as {\em{transfer matrix formalism}}, where the transfer matrix is given by $B^E(x)$.

Notice that since $\log\vert c \vert \in L^1(X,\ud \mu)$, the set $\mathcal{Z}(c):=\{x \in X : c(x) = 0\}$ is necessarily of $\mu$-measure zero. In particular, positivity of $\mu$ on open sets, implies that $(\mathrm{T}, B^E)$ is well-defined and invertible on the full measure, and therefore {\em{dense}}\footnote{In view of the extension of Theorem \ref{thm_main} to general uniquely ergodic systems as discussed in Remark \ref{rem_thm_main} (ii), density of $X_0$ forms the reason for requiring $\mu$ to be strictly positive on non-empty open sets of $X$.}.  $G_\delta$-set, 
\begin{equation} \label{eq_goodset}
X_0:=X \setminus \left( \cup_{n \in \mathbb{Z}} \mathrm{T}^n \mathcal{Z}(c) \right) ~\mbox{.}
\end{equation}

As $B^E(x)$ is only defined for $\mu$-a.e. $x$, it is often more convenient to work with $(\mathrm{T}, A^E)$, which inherits the continuity of the sampling functions $c,v$. We reiterate that presence of zeros in $c(x)$ translates to singularity of  $(\mathrm{T}, A^E)$. 

An alternative choice for the transfer matrix $B^E$ is given by,
\begin{equation} \label{eq_modified_B}
\widetilde{B}^E(x) = \dfrac{1}{c(\mathrm{T}^{-1} x)} \begin{pmatrix} E - v(x) & - \vert c(\mathrm{T}^{-1} x) \vert^2 \\ 1 & 0  \end{pmatrix} ~\mbox{,}
\end{equation}
which induces a measurable cocycle satisfying
\begin{equation}
\vert \det \widetilde{B}^E(x) \vert = 1 ~\mbox{, $\mu$-a.e.}
\end{equation}
We mention that for $E \in \mathbb{R}$, (\ref{eq_modified_B}) even induces a complex symplectic, measurable cocycle.

Thus, for {\em{non-singular}} QPJ where $(\mathrm{T}, \widetilde{B}^E(x))$ is {\em{continuous}}, dynamical considerations reduce directly to the more familiar notion of uniform hyperbolicity. The latter is explored in Sec. \ref{sec_nonsingularop}.

The definition of $\widetilde{B}^E(x)$ is suggested by the ``scaled'' discrete Laplacian in (\ref{eq_jacobiop}) \cite{PasturFigotin_book, DamanikKillipSimon_2010}; more specifically, $\psi \in \mathbb{C}^\mathbb{Z}$ satisfies $H_x \psi = E \psi$ if and only if 
\begin{equation} \label{eq_modified_B_1}
\widetilde{B}_n^E(x)  \begin{pmatrix} c(\mathrm{T}^{-1} x) \psi_0 \\ \psi_{-1}  \end{pmatrix} = \begin{pmatrix} c(\mathrm{T}^{n-1} x) \psi_n \\ \psi_{n-1}  \end{pmatrix} ~\mbox{.}
\end{equation}
We mention that $\widetilde{B}^E$ is particularly natural in view of the Weyl m-function $m_-(x,E)$, cf (\ref{eq_defstildeminus}) \footnote{We mention that the transfer matrix proposed in \cite{PasturFigotin_book, DamanikKillipSimon_2010} is in fact adapted to the Weyl m-function $m_+(x,E)$, hence differs from (\ref{eq_modified_B}). In view of proving presence of a dominated splitting, it is however more natural to adapt the cocycle to $m_-(x,E)$, as the latter will be shown to give rise to the dominating section (cf Sec. \ref{sec_domination}). We mention that all arguments in this note carry over to the cocyclces considered in \cite{PasturFigotin_book, DamanikKillipSimon_2010}, in particular Theorem \ref{thm_main} also applies to those cocycles.}.

For {\em{singular}} QPJ, however, as $\widetilde{B}^E(x)$ is only defined for $\mu$-a.e. $x$, one introduces in analogy to $A^E$ above, 
\begin{equation} \label{eq_modified_A}
\widetilde{A}^E(x) :=  \begin{pmatrix} E - v(x) & - \vert c(\mathrm{T}^{-1} x) \vert^2 \\ 1 & 0  \end{pmatrix} ~\mbox{,}
\end{equation}
which induces the (continuous) cocycle derived from $(\mathrm{T},\widetilde{B}^E)$. $(\mathrm{T}, \widetilde{A}^E)$ thus yields an alternative Jacobi cocycle, which, like $(\mathrm{T}, A^E)$, is singular precisely if $c$ exhibits zeros.

Referring to (\ref{eq_modified_B_1}), observe that the dynamics of the cocycles $(\mathrm{T}, A^E)$ and $(\mathrm{T}, \widetilde{A}^E)$ is related by the measurable conjugacy\footnote{As usual, a {\em{measurable conjugacy}} is a conjugacy between measurable cocycles where the mediating coordinate change in (\ref{eq_conjugacy_modcocycles}) is measurable with $\log \Vert M(.) \Vert , \log \Vert M(.)^{-1} \Vert \in L^1(X,\ud \mu)$. The latter condition guarantees preservation of the LE. Note that log-integrability of the coordinate change under consideration in (\ref{eq_conjugacy_modcocycles}) follows from $\log\vert c \vert \in L^1(X,\ud \mu)$.},
\begin{eqnarray} \label{eq_conjugacy_modcocycles}
\widetilde{M}(\mathrm{T} x)^{-1} \widetilde{A}^E(x) \widetilde{M}(x) = A^E(x) ~\mbox{, } 
\widetilde{M}(x):= \begin{pmatrix} 1 & 0 \\ 0 & \frac{1}{c(\mathrm{T}^{-1} x)} \end{pmatrix} ~\mbox{,}
\end{eqnarray}
which in particular yields equality of the top Lyapunov exponents,
\begin{equation} \label{eq_relLEs}
L(\mathrm{T}, A^E) = L(\mathrm{T}, \widetilde{A}^E) ~\mbox{, } L(\mathrm{T}, B^E) = L(\mathrm{T}, A^E) - \int \log \vert c \vert \ud \mu = L(\mathrm{T},\widetilde{B}^E) ~\mbox{.}
\end{equation}
We mention that in spectral theory, $L(\mathrm{T}, B^E)=L(\mathrm{T},\widetilde{B}^E)$ is usually called {\em{the}} Lyapunov exponent of the QPJ. 

For {\em{non-singular}} Jacobi operators, (\ref{eq_conjugacy_modcocycles}) becomes a {\em{continuous}} conjugacy, whence $(\mathrm{T}, A^E) \in \mathcal{DS}$ if and only if $(\mathrm{T}, \widetilde{A}^E) \in \mathcal{DS}$.

In conclusion, we account for the flexibility in the choice of cocycles associated with the spectral theory of QPJ, formulating our main result also for $(\mathrm{T}, \widetilde{A}^E)$: 
\begin{theorem}[Theorem \ref{thm_main} ammended] \label{thm_main_ammended}
Under the same hypotheses as in Theorem \ref{thm_main}, one has
\begin{eqnarray}
\Sigma = \{ E : (\mathrm{T}, A^E) \not \in \mathcal{DS} \} ~\mbox{.} \label{eq_thm_1a} 
\end{eqnarray} 
The statement also holds when replacing $(\mathrm{T}, A^E)$ by $(\mathrm{T}, \widetilde{A}^E)$.
\end{theorem}

\section{Non-singular Jacobi operators} \label{sec_nonsingularop}

As mentioned earlier, the spectral theory for {\em{non}}-singular Jacobi operators can be described by {\em{continuous}} cocycles with unimodular determinant similar to the Schr\"odinger case. In particular, since $(\mathrm{T}, A^E)$ are $(\mathrm{T}, \widetilde{A}^E)$ are continuously conjugate, Theorem \ref{thm_main} is equivalently formulated as
\begin{equation} \label{eq_thm_main_nonsing}
\mathbb{C}\setminus \Sigma = \{E: (\mathrm{T},\widetilde{B}^E) \in \mathcal{UH} \} ~\mbox{.}
\end{equation}
Thus, Theorem \ref{thm_main} can be concluded from arguments along the lines of \cite{Johnson_1986}. We briefly outline these straightforward adaptations.

To prove the ``$\subseteq$'' statement in (\ref{eq_thm_main_nonsing}), let $E \in \mathbb{C}\setminus \Sigma$ be given. It is well known that \cite{SackerSell_1974, SackerSell_1976, SackerSell_1978, Selgrade_1975} (see also, \cite{Zhang_notes_2013} for a more recent proof) a (continuous) cocycle $(\mathrm{T}, D)$ is {\em{not}} $\mathcal{UH}$ if and only if for {\em{some}} $x_0 \in X$ and $v \in \mathbb{C}^2 \setminus \{0\}$,
\begin{equation}
\sup_{n \in \mathbb{Z}} \Vert D_n(x_0) v \Vert < \infty ~\mbox{.}
\end{equation} 
Thus, using (\ref{eq_modified_B_1}), $(\mathrm{T}, \widetilde{B}^E) \not \in \mathcal{UH}$ would imply that for some $x_0 \in X$, $H_{x_0} \psi = E \psi$ admits a bounded, non-trivial solution over $\mathbb{C}^\mathbb{Z}$, whence (see also \ref{eq_schnol}) $E \in \Sigma$ - a contradiction.

To prove the ``$\supseteq$'' statement in (\ref{eq_thm_main_nonsing}), let $E$ such that $(\mathrm{T},\widetilde{B}^E) \in \mathcal{UH}$. Then, all non-trivial solutions of $H_x \psi = E \psi$ over $\mathbb{C}^\mathbb{Z}$ increase exponentially on at least one of $\mathbb{Z}_\pm$. Thus, the Sch'nol-Berezanskii theorem \cite{Schnol_1957, Berezanskii_1968} (see (\ref{eq_schnol})) and openness of $\mathcal{UH}$ in the (continuous) cocycles with unimodular determinant implies $E \in \mathbb{C}\setminus \Sigma$ (cf Sec. \ref{sec_finishingup}).

\section{Domination outside the spectrum} \label{sec_domination}
In this section we prove the ``$\supseteq$''-statement in (\ref{eq_thm_1a}). We start with some remarks on dominated splittings. 
\subsection{Dominated splittings and cone conditions} \label{subsec_domsplit}
For $1\leq j \leq 2$, let $e_j$ denote the standard basis of $\mathbb{C}^2$ and set $E_j:=\mathrm{Span} \{e_j\}$. Taking advantage of the manifold structure of $\mathbb{PC}^2$ induced by the charts,
\begin{eqnarray} \label{eq_charts}
\phi_1: \mathbb{PC}^2 \setminus E_1 \to \mathbb{C} ~\mbox{, } \phi_1(\mathrm{Span}\{ (\begin{smallmatrix} v_1 \\ v_2 \end{smallmatrix}) \} ) = \frac{v_1}{v_2} ~\mbox{,} \nonumber \\
\phi_2: \mathbb{PC}^2 \setminus E_2 \to \mathbb{C} ~\mbox{, } \phi_2(\mathrm{Span}\{ (\begin{smallmatrix} v_1 \\ v_2 \end{smallmatrix}) \} ) = \frac{v_2}{v_1} ~\mbox{,}
\end{eqnarray}
in local coordinates a given matrix $D = (\begin{smallmatrix} a & b \\ c & d \end{smallmatrix}) \neq 0$ acts on $\mathbb{PC}^2 \setminus \ker D$ as a linear fractional transformation. 

We will denote the coordinate-free action of $D$ on $\mathbb{PC}^2 \setminus \ker D$ by $D \cdot z$ and its derivative, respectively, by $\partial D \cdot z$. Moreover, we will find it convenient to identify $\mathbb{PC}^2$ with $\overline{\mathbb{C}}:=\mathbb{C} \cup \{\infty\}$ extending $\phi_2$ in (\ref{eq_charts}) to all of $\mathbb{PC}^2$.

First, observe that a cocycle $(\mathrm{T}, D) \in \mathcal{DS}$ if and only if {\em{some}} iterate is continuously conjugate to a diagonal cocycle, i.e. there exists $N \in \mathbb{N}$ and a coordinate change $M \in \mathcal{C}(X, GL(2,\mathbb{C}))$ such that
\begin{equation} \label{eq_ds_euiv}
M(\mathrm{T}^N x)^{-1} D_N(x) M(x) = \begin{pmatrix} \lambda_1(x) & 0 \\ 0 & \lambda_2(x) \end{pmatrix} ~\mbox{,}
\end{equation}
where $\lambda_j \in \mathcal{C}(X,\mathbb{C})$, $1 \leq j \leq 2$, satisfy
\begin{equation} \label{eq_ds_euiv_1}
\vert \lambda_1(x) \vert > \vert \lambda_2(x) \vert ~\mbox{, all $x \in X$ .}
\end{equation}

In particular, a necessary condition for $(\mathrm{T}, D) \in \mathcal{DS}$ is a non-degenerate Lyapunov spectrum, i.e.
\begin{equation} \label{eq_nondegenerate}
L(\mathrm{T}, D) > \frac{1}{2} \int \log \vert \det D(x) \vert \ud \mu(x) ~\mbox{.}
\end{equation}

A well-known technique to detect $\mathcal{DS}$ is to verify a {\em{cone condition}}: Given an $M_2(\mathbb{C})$-cocycle $(\mathrm{T},D)$, a conefield for $(\mathrm{T},D)$ is an open subset $U \subset X \times \mathbb{PC}^2$ of the form $\cup_{x \in X} \{x\} \times U_x$ such that, for all $x \in X$, $\overline{U_x}$ is non-empty, properly contained in $\mathbb{PC}^2$, and $\overline{U_x} \cap \ker D(x) = \emptyset$. A conefield $U = \cup_{x \in X} \{x\} \times U_x$ for $(\mathrm{T},D)$ is said to satisfy a cone condition if there exists $N \in \mathbb{N}$ such that for every $x \in X$, one can show that $D_N(x) \cdot \overline{U_x} \subseteq U_{\mathrm{T}^N x}$. It is known (see e.g. \cite{Avila_2011}, or \cite{AvilaJitomirskayaSadel_2013} for singular matrix cocycles) that verifying a cone condition implies $\mathcal{DS}$. In particular, note that if $s_1(x)$ denotes the dominating section for $(\mathrm{T},D)$ and $N \in \mathbb{N}$ is as in the cone condition, then one necessarily has
\begin{equation} \label{eq_dominationcrit}
\sup_{x \in X} \vert \partial D_N(x) \cdot s_1(x) \vert < 1 ~\mbox{.}
\end{equation}

Conversely, suppose for some $N \in \mathbb{N}$ a given invariant section $s(.) \in \mathcal{C}(X,\overline{\mathbb{C}})$, transverse to $\mathrm{ker} D$, satisfies the uniform contraction condition (\ref{eq_dominationcrit}), then $(\mathrm{T},D)$ admits a conefield satisfying a cone condition, whence $(\mathrm{T},D) \in \mathcal{DS}$ with $s(x)$ determining the dominating section.

\subsection{Proof of Theorem  \ref{thm_main_ammended}, ``$\supseteq$''-statement} \label{subsec_conecond}

Set $\rho:=\mathbb{C} \setminus \Sigma$. Throughout this section, let $E \in \rho$ be fixed. We will prove the ``$\supseteq$''-statement in (\ref{eq_thm_1a}) by showing that if $E \in \rho$, then both $(\mathrm{T},A^E)$ and $(\mathrm{T},\widetilde{A}^E)$ admit an invariant section such that (\ref{eq_dominationcrit}) holds for some $N \in \mathbb{N}$. As discussed in the end of Sec. \ref{subsec_domsplit}, the latter implies that both Jacobi cocycles induce a $\mathcal{DS}$. Here, as we will argue, the dominating, invariant section is provided from the spectral theory of Jacobi operators. 

To this end, let $m_\pm(x, E)$ denote the standard {\em{Weyl m-functions}} \cite{Teschl_book_2000}, 
\begin{equation} \label{eq_weyl}
m_\pm(x,E):= \langle \delta_{\pm 1} , (H_{x, \pm} - E)^{-1} \delta_{\pm 1} \rangle ~\mbox{,}
\end{equation}
defined in terms of the positive (negative) half-line operator $H_{x,\pm}:= P_\pm H_x P_\pm$, where $P_\pm$ is the orthogonal projection onto the closure of $\mathrm{Span}\{ \delta_n ~\mbox{, } n \in \mathbb{Z}_\pm \}$ and $\delta_n$, $n \in \mathbb{Z}$, are the elements of the standard basis of $\mathit{l}^2(\mathbb{Z})$, i.e. $[\delta_n]_m = \delta_{n,m}$.

In view of $(\mathrm{T},A^E)$, we consider
\begin{eqnarray}   \label{eq_defsminus}
s_-(x,E) & := & -c(T^{-1} x) m_-(x,E)  ~\mbox{, } \\
s_+(x,E) & := & -\{ \overline{c(T^{-1} x)} m_+(T^{-1} x , E) \}^{-1}  \nonumber ~\mbox{,}
\end{eqnarray}
and for $(\mathrm{T},\widetilde{A}^E)$, respectively,
\begin{eqnarray}  \label{eq_defstildeminus}
\widetilde{s}_-(x,E) & := & - m_-(x,E)  ~\mbox{, } \\
\widetilde{s}_+(x,E) & := & - \{\vert c(T^{-1} x)\vert^2 m_+(T^{-1} x , E) \}^{-1} \nonumber ~\mbox{.}
\end{eqnarray}
Clearly, for all $E \in \mathbb{C}\setminus\mathbb{R}$,
\begin{equation} \label{eq_mfuncboundscomplex}
0 < \vert m_\pm(x,E) \vert \leq \dfrac{1}{\vert \im E \vert} ~\mbox{,}
\end{equation}
thus $s_\pm(x,E), \widetilde{s}_\pm(x,E)$ are well-defined with values in $\overline{\mathbb{C}}$.

To see that (\ref{eq_defsminus}) and (\ref{eq_defstildeminus}) are actually well defined for all $E \in \rho$ (i.e. undefined expressions of the form ``$0 \times \infty$'' do not occur), first observe that $H_x$ and $H_{x,-} \oplus H_{x,+}$ only differ by a finite rank perturbation, hence their essential spectra must agree. In particular, any real $E$ in the resolvent set of $H_x$ is either in the resolvent sets of both $H_{x,\pm}$, i.e. $m_\pm(x,E) \in \mathbb{C}$ ,\footnote{In contrast to $H_x$, the resolvent sets for $H_{x,\pm}$ may in general depend on $x$.} or in the discrete spectrum\footnote{As usual, the {\em{discrete spectrum}} of a bounded operator is defined as the set of {\em{isolated}} eigenvalues of {\em{finite}} multiplicity; the remaining elements of the spectrum define the {\em{essential spectrum}}.} of at least one of $H_{x,\pm}$, i.e. $m_\pm(x,E) = \infty$ corresponding to a pole at $E$. Thus, undefined expressions in (\ref{eq_defsminus})-(\ref{eq_defstildeminus}) of the form ``$0 \times \infty$'' are excluded: indeed, as $c(\mathrm{T}^{-1} x) = 0$ implies $H_x = H_{x,-} \oplus H_{\mathrm{T}^{-1} x,+}$, any $E \in \sigma_{\mathrm{disc}}(H_{x,-}) \cup \sigma_{\mathrm{disc}}(H_{\mathrm{T}^{-1} x,+})$ would automatically be an eigenvalue of $H_x$, thereby violating our assumption that $E \in \rho$.

We claim:
\begin{lemma} \label{lem_contisections}
\begin{equation} \label{eq_conti}
s_\pm(. , .), \widetilde{s}_\pm(. , .) \in \mathcal{C}(X \times \rho,\overline{\mathbb{C}}) ~\mbox{.}
\end{equation}
\end{lemma}
\begin{proof}
It suffices to show joint continuity of $m_\pm(x,E)$. Let $(x_0,E_0) \in X \times \rho$ be fixed and arbitrary. As outlined above, there are two possible situations to consider. 

If $E_0$ is in the resolvent set of $H_{x_0,\pm}$, basic resolvent estimates (see e.g. Theorem 3.15 in \cite{Kato_book}) imply that $(H_{x,\pm} - E)^{-1}$ exists and is jointly continuous in $(x,E)$ (w.r.t operator norm for the $x$-dependence \footnote{In view of Remark \ref{rem_thm_main} (ii), we mention the hypothesis of Theorem 3.15 in \cite{Kato_book} requires that $x \mapsto H_x$ be continuous in operator norm.}) in some open neighborhood of $(x_0,E_0)$; in particular, $m_\pm(x,E)$ is jointly continuous at $(x_0,E_0)$. 

Continuity for the case that $E_0$ is in the discrete spectrum of $H_{x_0,\pm}$ follows from well-known facts on the continuity of a finite system of isolated eigenvalues of finite multiplicity for a norm-continuous family of bounded operators (see e.g. Sec. IV.5 in \cite{Kato_book}): 
Specifically, let 
\begin{equation}
P_{\sigma_1(H_{x_0, \pm})}(x_0):= \frac{1}{2 \pi i} \int_{\partial B_r(E_0)} (H_{x_0, \pm} - z)^{-1} \ud z ~\mbox{,}
\end{equation}
be the spectral projection of $H_{x_0, \pm}$ onto $\{ E_0 \}$, where $r>0$ is such that  $\overline{B_r(E_0)} \cap \sigma(H_{x_0, \pm}) = \{E_0 \}$. 

Then, from Sec. IV.5 in \cite{Kato_book} (see also \cite{Simon_analyis_operatortheory}, Theorem 2.3.8), it is known that there exists $\delta>0$ such that for all $x$ with $\vert x - x_0 \vert < \delta$, one has:
\begin{itemize}
\item[(i)] $\sigma(H_{x, \pm})$ can be partioned into the compact sets $\sigma_1(H_{x, \pm})$ and $\sigma_2(H_{x, \pm})$ such that $\sigma_1(H_{x, \pm}) \subseteq B_r(E_0)$ and $\sigma_2(H_{x, \pm}) \subseteq \mathbb{C} \setminus \overline{B_r(E_0)}$.
\item[(ii)] The spectral projections 
\begin{equation}
P_{\sigma_1(H_{x, \pm})}(x):= \frac{1}{2 \pi i} \int_{\partial B_r(E_0)} (H_{x, \pm} - z)^{-1} \ud z ~\mbox{,}
\end{equation}
depend continuously (in norm topology) on $x$ and are unitarily equivalent to $P_{\sigma_1(H_{x_0, \pm})}(x_0)$. In particular, the spectrum of $H_{x, \pm}$ inside $B_r(E_0)$ is discrete with dimension given by $1 \leq \dim \mathrm{Ran} P_{\sigma_1(H_{x_0, \pm})}(x_0) < \infty$.
\end{itemize}

Thus, using (i)-(ii), for $\vert x - x_0 \vert < \delta$, decompose 
\begin{equation}
H_{x, \pm} = H_{x, \pm}^{(1)} \oplus H_{x, \pm}^{(2)} ~\mbox{,}
\end{equation}
where $H_{x, \pm}^{(1)}:= H_{x, \pm} P_{\sigma_1(H_{x, \pm})}(x)$ and $H_{x, \pm}^{(2)}:= H_{x, \pm} (I - P_{\sigma_1(H_{x, \pm})}(x))$. Then, all $E \in B_r(E_0)$ are in the resolvent set of $H_{x, \pm}^{(2)}$, in particular, as above, $(H_{x, \pm}^{(2)} - E)^{-1}$ is jointly continuous in $(x,E) \in (x_0 - \delta, x_0 + \delta) \times B_r(E_0)$.

Therefore, joint continuity of $m_\pm(x, E)$ at $(x_0,E_0)$ as $\overline{\mathbb{C}}$-valued functions, simply reduces to the continuity of the $\overline{\mathbb{C}}$-valued functions,
\begin{equation} \label{eq_katolemma}
\langle \delta_{\pm 1} , (H_{x, \pm}^{(1)} - E)^{-1} \delta_{\pm 1} \rangle ~\mbox{, for } E \in B_r(E_0) ~\mbox{, } \vert x - x_0 \vert < \delta ~\mbox{,}
\end{equation}
associated with the {\em{finite dimensional}} operators $H_{x, \pm}^{(1)}$. Here, (\ref{eq_katolemma}) has poles at the eigenvalues of $H_{x, \pm}^{(1)}$ for $E \in B_r(E_0)$.
\end{proof}

Recall, that the definition of the full-measure set $X_0 \subseteq X$ given in (\ref{eq_goodset}), guarantees that $c(\mathrm{T}^k x) \neq 0$, $\forall k \in \mathbb{Z}$, whenever $x \in X_0$. For all $x \in X_0$, this in turn implies existence of solutions $\psi_\pm(x,E)$ of $H_x \psi = E \psi$ over $\mathbb{C}^\mathbb{Z}$ which are never zero, are $\mathit{l}^2$ at $\pm \infty$, and from (\ref{eq_B}), respectively (\ref{eq_modified_B}), one has
\begin{equation}
s_\pm(x,E) = \dfrac{ \psi_\pm(-1,x,E)}{\psi_\pm(0,x,E)} ~\mbox{, } \widetilde{s}_\pm(x,E) = \dfrac{ \psi_\pm(-1,x,E)}{c(\mathrm{T}^{-1} x) \psi_\pm(0,x,E)} ~\mbox{, $x \in X_0$,}
\end{equation}
and
\begin{eqnarray} \label{eq_invariance_2}
- m_-(\mathrm{T} x,E)^{-1} = \begin{cases} (E-v(x)) - \overline{c(\mathrm{T}^{-1} x)} s_-(x,E) ~\mbox{,} \\ 
                                                                            (E-v(x)) - \vert c(\mathrm{T}^{-1} x)\vert^2 \widetilde{s}_-(x,E) ~\mbox{.} \end{cases}                                                                          
\end{eqnarray}

Moreover, uniqueness of the fundamental solutions $\psi_\pm(x,E)$ up to scalar multiples implies that for all $x \in X_0$,
\begin{equation} \label{eq_invariance}
A^E(x) \cdot s_\pm(x,E) = s_\pm(\mathrm{T} x,E) ~\mbox{, } 
\widetilde{A}^E(x) \cdot \widetilde{s}_\pm(x,E) = \widetilde{s}_\pm(\mathrm{T} x,E) ~\mbox{.}
\end{equation}

We emphasize that outside $X_0$ above fundamental solutions $\psi_\pm(x,E)$ do not exist; indeed, $c(\mathrm{T}^k x) =0$ for some $k \in \mathbb{Z}$, implies decoupling of $H_x$ whence any solution $\psi$ of $H_x \psi = E \psi$ which is $\mathit{l}^2$ at $\pm \infty$ has to vanish {\em{identically}} in a neighborhood of $\pm \infty$ otherwise $E$ is an eigenvalue of $H_x$. 

In order to extend (\ref{eq_invariance_2}) and the invariance relations in (\ref{eq_invariance}) to all of $X$, we use the following continuity arguments: Since $X_0$ is dense in $X$, continuity of both sides of (\ref{eq_invariance_2}) as $\overline{\mathbb{C}}$-valued functions implies that (\ref{eq_invariance_2}) extends to all of $X$.  Moreover, one has:
\begin{lemma} \label{lem_transversality}
For all $E \in \rho$ and all $x \in X$, $s_-(x,E)$ is transverse to $\ker A^E(x)$, similarly $\widetilde{s}_-(x,E)$ is transverse to $\ker \widetilde{A}^E(x)$. Moreover, for all $x \in X$
\begin{equation}
A^E(x) \cdot s_-(x,E) = s_-(\mathrm{T} x,E) ~\mbox{, } 
\widetilde{A}^E(x) \cdot \widetilde{s}_-(x,E) = \widetilde{s}_-(\mathrm{T} x,E) ~\mbox{.}
\end{equation}
\end{lemma}

\begin{proof}
We will focus on $(\mathrm{T}, A^E)$, the argument for $(\mathrm{T}, \widetilde{A}^E)$ being similar. 

First observe that for $E \in \rho$, $A^E(x) \not \equiv 0$ since $c(x) = c(\mathrm{T}^{-1}x)=0$ would require $E - v(x) \neq 0$ for $E$ to be in the resolvent set of $H_x$. In particular, $\dim \ker A^E(x) \leq 1$. 

Based on (\ref{eq_defA}), $\ker A^E(x)$ is non-trivial if one of the following two situations applies:

If $c(\mathrm{T}^{-1}x) = 0$, $s_-(x,E) = 0$ which is automatically transverse to $\ker A^E(x) = E_2 \simeq \infty$.

If $c(x)=0$, we may assume $c(\mathrm{T}^{-1}x) \neq 0$, otherwise consider above. Then, $s_-(x,E) \simeq \mathrm{Span} (\begin{smallmatrix} 1 \\ s_-(x,E) \end{smallmatrix}) = \ker A^E(x)$ if and only if
\begin{equation}
(E-v(x)) - \overline{c(\mathrm{T}^{-1}x)} s_-(x,E) = 0 ~\mbox{,}
\end{equation}
which as (\ref{eq_invariance_2}) holds on all of $X$ would imply a pole of $m_-(\mathrm{T} x, . )$ a $E$, or equivalently, $E \in \sigma_{\mathrm{disc}}( H_{\mathrm{T}x,-} )$. The latter however is impossible for any $E \in \rho$ as $c(x) = 0$ yields $H_x = H_{\mathrm{T}x,-} \oplus H_{x,+}$.

Finally, transversality of $s_-(x,E)$ to $\ker A^E(x)$ for all $x \in X$, implies that \newline $A^E(.) \cdot s_-(.,E)$ is well defined and continuous with values in $\overline{\mathbb{C}}$, whence density of $X_0$ allows to extend (\ref{eq_invariance}) to all of $X$. 
\end{proof}

Based on the discussion in Sec. \ref{subsec_domsplit}, we show that $(\mathrm{T}, A^E), (\mathrm{T}, \widetilde{A}^E) \in \mathcal{DS}$ verifying the contraction condition given in (\ref{eq_dominationcrit}). Recall that by the Combes-Thomas estimate (see e.g. \cite{Teschl_book_2000}, Lemma 2.5, for a formulation for Jacobi operators), 
\begin{equation} \label{eq_combesthomas}
L(\mathrm{T}, B^E) \geq \kappa \cdot \mathrm{dist}(E;\Sigma) > 0 ~\mbox{, $E \in \rho$,}
\end{equation}
where, uniformly over any compact neighborhood of $\Sigma$, $\kappa>0$ can be chosen to only depend on $\Vert c \Vert_\infty$. In particular, using (\ref{eq_nondegenerate}) and (\ref{eq_relLEs}), (\ref{eq_combesthomas}) shows that for any $E \in \rho$ both $(\mathrm{T}, A^E), (\mathrm{T}, \widetilde{A}^E)$ have a non-degenerate Lyapunov spectrum.

In view of the following, we let
\begin{equation} \label{eq_defrhominus}
\rho_-:=\mathbb{C}\setminus ( \Sigma \cup \left(  \cup_{x \in X} \sigma_{\mathrm{disc}}(H_{x,-}) \right)) ~\mbox{.}
\end{equation}
Clearly, $\mathbb{C} \setminus \mathbb{R} \subseteq \rho_-$, moreover $E \in \rho_-$ if $\vert E \vert \geq  2 \Vert c \Vert_\infty + \Vert v \Vert_\infty$.
\begin{prop} \label{prop_unifcontractcond}
For any $E \in \rho_-$ one has
\begin{equation} \label{eq_estimderiv}
\sup_{x \in X} \left\vert \partial A_n^E(x) \cdot s_-(x,E) \right\vert \leq \Vert c \Vert_\infty^2 ~ \mathrm{e}^{-2 (L(\mathrm{T},A^E) + o(1))} \mathrm{e}^{- 2 (n-1) ( L(\mathrm{T},B^E) + o(1) )}  ~\mbox{,}
\end{equation}
as $n \to +\infty$. An analogous estimate holds for $\widetilde{A}^E(x)$.
\end{prop}
\begin{remark}
The proof also shows that the upper bound in (\ref{eq_estimderiv}) is optimal. Thus, as $L(\mathrm{T},A^E) \geq \int \log\vert c\vert \ud \mu$, $N$ in the definition of $\mathcal{DS}$ (see (\ref{eq_defDS})) will diverge whenever $L(\mathrm{T},B^E) \to 0$ as $\mathrm{dist}(E;\Sigma) \to 0+$.
\end{remark}

\begin{proof}
For brevity, we will focus on the cocycle $(\mathrm{T}, A^E)$. First observe that for all $x \in X$, $s_-(x,E), m_-(\mathrm{T}x,E) \neq \infty$ as $E \in \rho_-$, whence from (\ref{eq_invariance_2}) we conclude
\begin{equation}
A^E(x) \cdot z = (\phi_2 \circ A^E \circ \phi_2^{-1})(z) = \dfrac{c(x)}{(E-v(x)) - c(\mathrm{T}^{-1} x) z} ~\mbox{,}
\end{equation}
{\em{locally}} about $z=s_-(x,E)$ for all $x \in X$.

Thus, again using (\ref{eq_invariance_2}), we compute
\begin{eqnarray} \label{eq_deriv}
\partial A^E(x) \cdot s_-(x,E) & = & \dfrac{ c(x) \overline{c(\mathrm{T}^{-1} x)}}{\left( \left(E-v(x)\right) - \overline{c(\mathrm{T}^{-1} x)} s_-(x,E) \right)^2} \\
                                              & = & c(x) \overline{c(\mathrm{T}^{-1} x)} m_-^2(\mathrm{T} x, E) ~\mbox{.}
\end{eqnarray}

From (\ref{eq_invariance}), (\ref{eq_deriv}), the chain rule implies
\begin{eqnarray} \label{eq_deriv_1}
\partial A_n^E(x) \cdot s_-(x,E) 
                                                   & = & \prod_{j=0}^{n-1} \partial A^E(\mathrm{T}^j x) \cdot  \left( A_{j}^E(x) \cdot s_-(x,E) \right) \nonumber \\
                                                   & = & \left( \prod_{j=0}^{n-1} c(\mathrm{T}^j x)   \right) \left( \prod_{j=-1}^{n-2} \overline{c(\mathrm{T}^j x)} \right) \left( \prod_{j=1}^{n} m_-^2(\mathrm{T}^j x , E) \right) ~\mbox{, } n \geq 2 ~\mbox{.}                                                   
\end{eqnarray}

Relating $m_-(x,E)$ to the fundamental solution $\psi_-(x,E)$, one determines, making use of ergodicity (see e.g. \cite{Teschl_book_2000}, Eq. (5.40) therein), that
\begin{equation} \label{eq_limit1}
\lim_{n \to +\infty} \frac{1}{n} \sum_{j=0}^{n-1} \log \vert m_-(\mathrm{T}^j x,E) \vert = - L(\mathrm{T}, A^E) ~\mbox{, a.e. $x \in X$.}
\end{equation}

Observe that since $m_-(x,E)$ is continuous with 
\begin{equation} 
0 \leq \vert m_-(x,E) \vert < \infty ~\mbox{, $E \in \rho_-$ ,}
\end{equation}
unique ergodicity of $\mathrm{T}$ implies uniformity of the {\em{upper}} limit in (\ref{eq_limit1})  \cite{Furman_1997} \footnote{To obtain (\ref{eq_limit1a}) and (\ref{eq_limit2}), we use Theorem 1 in \cite{Furman_1997} which guarantees uniform convergence of {\em{upper}} limits in Caes\`aro means for continuous, sub-additive processes on a compact Hausdorff space $X$ equipped with a uniquely ergodic dynamical system. The proof in \cite{Furman_1997} carries over without changes to sub-additive processes $\{f_n\}$ which are only {\em{upper semi-continuous}}, thus in particular satisfy $\sup_{x \in X} f_1(x) < \infty$, which is used in the proof of\cite{Furman_1997}.
We mention that recently, Furman's result has been extended to even encompass certain discontinuous processes \cite{MaviJitomirskaya_2012}.}
\begin{equation} \label{eq_limit1a}
\frac{1}{n} \sum_{j=0}^{n-1} \log \vert m_-(\mathrm{T}^j x,E) \vert \leq - L(\mathrm{T}, A^E) + o(1) ~\mbox{, {\em{uniformly}} for $x \in X$.}
\end{equation}

The same argument yields
\begin{equation} \label{eq_limit2}
\frac{1}{n} \sum_{j=0}^{n-1} \log \vert c(\mathrm{T}^j x) \vert \leq \int \log \vert c(x) \vert \ud \mu(x) + o(1) ~\mbox{,}
\end{equation}
uniformly in $x \in X$ as $n \to +\infty$. 

Thus, combination of (\ref{eq_relLEs}), (\ref{eq_limit1a}), (\ref{eq_limit2}), and (\ref{eq_deriv_1}) yields (\ref{eq_estimderiv}) as claimed.
\end{proof}

Thus, letting $S_\pm(.,E)$ be continuous lifts of $s_\pm(.,E)$ to the subspaces of $\mathbb{C}^2$, Proposition \ref{prop_unifcontractcond} and (\ref{eq_invariance}) implies a $\mathcal{DS}$ for all $E \in \rho_-$, with $S_-(x,E)$ corresponding to the dominating subspace, in particular
\begin{equation}
A^E(x) S_-(x,E) = S_-(\mathrm{T} x,E) ~\mbox{.} \label{eq_invdom}
\end{equation}
That $S_+(x,E)$ determines the minorating subspace of the dominated splitting follows by first noting that from (\ref{eq_invariance}) one has for all $x \in X_0$,
\begin{eqnarray}
A^E(x) S_+(x,E) \subseteq S_+(\mathrm{T} x,E) ~\mbox{.} \label{eq_invmin} 
\end{eqnarray}
Moreover, one has:
\begin{lemma} \label{lem_transversality_2}
For all $E \in \rho$, $S_\pm(.,E)$ are uniformly transverse.
\end{lemma}
\begin{proof}
For any $E \in \rho$ and $x \in X_0$, standard expressions for the Green's function of $H_x$ in terms of $\psi_\pm(x,E)$ yield
\begin{eqnarray} \label{eq_greensfunctransv_1}
\langle \delta_0, (H_x - E)^{-1} \delta_0 \rangle ^{-1} & = & c(x) \left( \dfrac{\psi_+(1,x,E)}{\psi_+(0,x,E)} - \dfrac{\psi_-(1,x,E)}{\psi_-(0,x,E)} \right) \nonumber \\
                                                                                              & = & \overline{c(\mathrm{T}^{-1} x)} \left( \dfrac{\psi_-(-1,x,E)}{\psi_-(0,x,E)} - \dfrac{\psi_+(-1,x,E)}{\psi_+(0,x,E)} \right) ~\mbox{,}
\end{eqnarray}
whence
\begin{eqnarray} \label{eq_greensfunctransv}
\left \vert s_+(x,E) - s_-(x,E) \right \vert \geq \dfrac{\mathrm{dist}(E;\Sigma)}{\Vert c \Vert_\infty} ~\mbox{, all $x \in X_0$ .} 
\end{eqnarray}

By (\ref{eq_mfuncboundscomplex}), $s_+(x,E)$, $s_-(x,E)$ can never {\em{both}} equal $\infty$ for $E \in \mathbb{C}\setminus\mathbb{R}$, whence (\ref{eq_greensfunctransv}) extends to all $x \in X$ by continuity. 

Given {\em{real}} $E \in \rho$, using (\ref{eq_greensfunctransv}) there exists $\eta>0$ only depending on $\mathrm{dist}(E;\Sigma)$ such that for all $\epsilon>0$,
\begin{equation}
\inf_{x \in X} \angle\{ S_+(x,E+i \epsilon) ~,~ S_-(x,E+i\epsilon) \} \geq \eta ~\mbox{,}
\end{equation}
which by continuity implies the claimed transversality taking $\epsilon \to 0+$.
\end{proof}

Proposition \ref{prop_unifcontractcond} was restricted to $E \in \rho_-$ since this guaranteed boundedness of $m_-(x,E)$, which was crucial to conclude uniformity of the upper limit in (\ref{eq_limit1a}). Nonetheless, having shown $(\mathrm{T}, A^E), ~(\mathrm{T}, \widetilde{A}^E)  \in \mathcal{DS}$ for all $E \in \mathbb{C}\setminus \mathbb{R}$, however already implies the same for all of $\rho$; we provide an argument for $(\mathrm{T}, A^E)$:

From Lemma \ref{lem_transversality}, (\ref{eq_invdom}) already holds for all $E \in \rho$. Moreover, using Lemma \ref{lem_contisections}, $S_+(.,E)$ extends continuously to all of $E \in \rho$, whence (\ref{eq_invmin}) holds for all $x \in X$. In summary, taking $M(x) \in M_2(\mathbb{C})$ with the first (second) column vector in the direction of $S_-(x,E)$ $(S_+(x,E))$, Lemma \ref{lem_transversality_2} implies that $M(x) \in GL(2,\mathbb{C})$ and
\begin{equation} \label{eq_conjds}
M(\mathrm{T} x)^{-1} A^E(x) M(x) = \begin{pmatrix} \lambda_1(x) & 0 \\ 0 & \lambda_2(x) \end{pmatrix} ~\mbox{, } 
\end{equation}
where $\lambda_j \in \mathcal{C}(X,\mathbb{C})$, $ 1 \leq j \leq 2$. Since $S_-(x,E)$ is transversal from $\ker A^E(x)$, one has $\lambda_1(x) \neq 0$, thus by continuity of $\lambda_1$,
\begin{equation} \label{eq_conjds1}
 \inf_{x \in X} \vert \lambda_1(x) \vert > 0 ~\mbox{.}
\end{equation}

Finally, non-triviality of the Lyapunov spectrum of $(\mathrm{T},A^E)$ for all $E \in \mathbb{C}\setminus\Sigma$, the equations (\ref{eq_conjds})-(\ref{eq_conjds1}), and unique ergodicity of $\mathrm{T}$ implies that for some $N \in \mathbb{N}$,
\begin{equation}
\left\vert \prod_{j=0}^{N} \lambda_1(\mathrm{T}^jx) \right\vert > \left\vert \prod_{j=0}^{N} \lambda_2(\mathrm{T}^jx) \right\vert ~\mbox{, all $x \in X$,}
\end{equation}
whence $(\mathrm{T},A^E) \in \mathcal{DS}$ (cf (\ref{eq_ds_euiv}) - (\ref{eq_ds_euiv_1})).

\section{Finishing up $\dots$} \label{sec_finishingup}
To complete the proof of Theorem \ref{thm_main_ammended}, it is left to show that $\mathcal{DS}$ of $(\mathrm{T}, A^E)$ or $(\mathrm{T}, \widetilde{A}^E)$ cannot occur on the spectrum. 

Suppose, for some $E \in \mathbb{C}$ one has $(\mathrm{T}, A^E) \in \mathcal{DS}$, correspondingly giving rise to a conjugacy of the form (\ref{eq_ds_euiv}), in particular
\begin{equation} \label{eq_johnson_1}
\det A_N^E(x) = \dfrac{\det M(\mathrm{T}^N x)}{ \det M(x) }\lambda_1(x) \lambda_2(x) ~\mbox{.}
\end{equation}
Similarly, if $E \in \mathbb{C}$ such that $(\mathrm{T}, \widetilde{A}^E) \in \mathcal{DS}$, (\ref{eq_johnson_1}) holds true with $A^E$ replaced by $\widetilde{A}^E$.

\begin{prop} \label{claim_johnson}
Suppose $E \in \mathbb{C}$ is such that $(\mathrm{T}, A^E) \in \mathcal{DS}$ or $(\mathrm{T}, \widetilde{A}^E) \in \mathcal{DS}$. Then, there is $\gamma>0$ such that for all $x \in X_0$ and $v \in \mathbb{C}^2 \setminus \{0\}$, there exists $c_{x,v} > 0$ and a subsequence $(k_l)_{l \in \mathbb{N}}$ of one of $\mathbb{Z}_\pm$ satisfying
\begin{equation} \label{eq_claimjohnson}
\Vert B^E_{k_{l}N}(x) v \Vert \geq c_{x,v} ~\mathrm{e}^{\vert k_l \vert \gamma} ~\mbox{, for all $l \in \mathbb{N}$.} 
\end{equation}
\end{prop}
\begin{proof}
For $E \in \mathbb{C}$ such that $(\mathrm{T}, A^E) \in \mathcal{DS}$, consider
\begin{equation}
(A^E)^\sharp(x):=\dfrac{A^E(x)}{\sqrt{ \vert \det A^E(x) \vert }} ~\mbox{,}
\end{equation}
i.e. on  $X_0 \times \mathbb{C}^2$, $(\mathrm{T}, (A^E)^\sharp)$ is well-defined and invertible for all iterates $n \in \mathbb{Z}$ with $\vert \det (A^E)^\sharp(x) \vert = 1$.

Let $x \in X_0$ be arbitrary and fixed, in particular by definition of the set $X_0$ in (\ref{eq_goodset}), one has
\begin{equation*}
c(T^n x) \neq 0 ~\mbox{, for all $n \in \mathbb{Z}$.} 
\end{equation*} 

Given $v \in \mathbb{C}^2 \setminus\{0\}$, decompose $v = v_1 + v_2$ with $v_j \in S_x^{(j)}$. Possibly changing to inverse dynamics, we may assume $v_1$ to be non-zero. From (\ref{eq_johnson_1}), one concludes that for $k \in \mathbb{N}$ 
\begin{equation} \label{eq_claimjohnson_1}
\Vert (A_{kN}^E)^\sharp(x) v_1 \Vert = \left\vert  \dfrac{\det M(x) }{\det M(\mathrm{T}^{kN} x)} \right\vert^{1/2} \Vert v_1 \Vert \prod_{j=0}^{k-1} \left\vert  \dfrac{ \lambda_1(\mathrm{T}^{jN} x)    }{ \lambda_2(\mathrm{T}^{jN} x)   }      \right\vert^{1/2} \geq {\tilde{c}}_{v} \Vert v_1 \Vert \mathrm{e}^{\frac{1}{2} k \log \lambda} ~\mbox{,}
\end{equation}
where $\lambda =  \inf_{x \in X} \frac{ \vert \lambda_1(x) \vert }{ \vert \lambda_2(x) \vert } > 1$ and ${\tilde{c}}_{v} = \frac{\inf_{x \in X} \vert \mathrm{det} M(x) \vert}{ \Vert \mathrm{det} M \Vert_\infty} >0$. Hence all non-trivial iterates of $(\mathrm{T}, (A^E)^\sharp)$ increase exponentially along a subsequence (depending both on $x$ and the initial condition $v$). An analogous argument shows that iterates of the component of $v$ along $v_2$ {\em{decay}} exponentially at rate at least $\lambda$.


To conclude the same for the respective solutions of $H_x \psi= E \psi$, observe that one has
\begin{equation} \label{eq_claimjohnson_2}
\Vert (A_{k N}^E)^\sharp(x) v \Vert = \dfrac{ \vert c(\mathrm{T}^{k N -1} x) \vert^{1/2} }{ \vert c(\mathrm{T}^{-1} x) \vert^{1/2}} \Vert B_{k N}^E(x) v \Vert \leq \left(\dfrac{\Vert c \Vert_\infty  }{ \vert c(\mathrm{T}^{-1} x) \vert } \right)^{1/2} \Vert B_{k N}^E(x) v \Vert~\mbox{,}
\end{equation}
which, combining the right most-sides of (\ref{eq_claimjohnson_1})--(\ref{eq_claimjohnson_2}), yields
\begin{equation}
\Vert B_{k N}^E(x) v \Vert \geq {\tilde{c}}_{v,x} \left(\dfrac{ \vert c(\mathrm{T}^{-1} x) \vert }{\Vert c \Vert_\infty  } \right)^{1/2} \Vert v \Vert \mathrm{e}^{\frac{1}{2} k \log \lambda} ~\mbox{,}
\end{equation}
eventually in $k$, as claimed in (\ref{eq_claimjohnson}).

The case when $(\mathrm{T}, \widetilde{A}^E) \in \mathcal{DS}$ follows along the same line, noticing that since $x \in X_0$, one has
\begin{equation}
\Vert (\widetilde{A}_n^E)^\sharp(x) v \Vert = \Vert \widetilde{B}_n^E(x) v \Vert ~\mbox{, $n \in \mathbb{Z}$,}
\end{equation}
which, in analogy to (\ref{eq_claimjohnson_1}), shows that eventually in $k$,
\begin{equation} \label{eq_claimjohnson_3a}
\Vert \widetilde{B}_{kN}^E(x) v \Vert \gtrsim \mathrm{e}^{\frac{1}{2} k \log \lambda} \Vert v \Vert ~\mbox{.}
\end{equation}
Thus, for $k$ sufficiently large, the conjugacy in (\ref{eq_conjugacy_modcocycles}) implies
\begin{equation} \label{eq_claimjohnson_3}
\mathrm{e}^{\frac{1}{2} k \log \lambda} \Vert \widetilde{M}(x) v \Vert \lesssim \Vert \widetilde{B}_{k N}^E(x) \widetilde{M}(x) v \Vert \lesssim  \dfrac{ \Vert \widetilde{M} \Vert_\infty \Vert c \Vert_\infty }{ \vert c(\mathrm{T}^{-1} x) \vert} \Vert B_{k N}^E(x) v \Vert ~\mbox{,}
\end{equation}
which completes the proof.
\end{proof}

For fixed $x \in X$, denote by $\mathcal{E}_g(H_x)$ the set of generalized eigenvalues of $H_x$, i.e. all $E \in \mathbb{C}$ which admit a non-trivial {\em{polynomially bounded}} solution of $H_x \psi = E \psi$ over $\mathbb{C}^\mathbb{Z}$. From the theorem of Sch'nol-Berezanskii \cite{Schnol_1957, Berezanskii_1968} it is well-known that
\begin{equation} \label{eq_schnol}
\overline{\mathcal{E}_g(H_x)} = \Sigma ~\mbox{,}
\end{equation}
for all $x \in X$.

Proposition \ref{claim_johnson} shows that for all $x \in X_0$, all non-trivial solutions to $H_x \psi = E \psi$ over $\mathbb{C}^\mathbb{Z}$ are {\em{not}} polynomially bounded, whence $E \not \in \mathcal{E}_g(H_x)$. But then, since $\mathcal{DS}$ is an open property in $\mathcal{C}(X,M_2(\mathbb{C}))$, we conclude that for all $x \in X_0$, $E$ cannot be a limit point of $\mathcal{E}_g(H_x)$ either. Thus, in summary, (\ref{eq_schnol}) implies the ``$\subseteq$''-statement of Theorem \ref{thm_main}.
Using Proposition \ref{claim_johnson} for $\widetilde{A}^E$, the same argument yields the ``$\subseteq$''-statement for the alternative Jacobi cocycle $(\mathrm{T},\widetilde{A}^E)$ in Theorem \ref{thm_main_ammended}.
 
\bibliographystyle{amsplain}

\end{document}